\documentclass[aps,prl,10pt,twocolumn,superscriptaddress,floatfix,nofootinbib,longbibliography]{revtex4-2}
\usepackage{amsmath,amssymb,amsthm,bm}
\usepackage{graphicx}
\usepackage{txfonts}
\usepackage{xcolor}
\usepackage[colorlinks=true,linkcolor=blue,citecolor=blue,urlcolor=blue]{hyperref}

\newtheorem{theorem}{Theorem}
\newtheorem{proposition}{Proposition}
\newtheorem{corollary}{Corollary}
\newcounter{smsec}
\renewcommand{\thesmsec}{S\arabic{smsec}}
\newcommand{\smsection}[2]{\refstepcounter{smsec}\section*{\thesmsec. #1}\label{#2}}
\newtheorem{smlemma}{Lemma}

\newtheorem{smprop}{Proposition}

\begin{document}

\title{Many-body entanglement quantification in arbitrary states through a  few partial-transpose moments}

\author{Sudip Chakrabarty}
\email{sudip27042000@gmail.com}
\affiliation{S. N. Bose National Centre for Basic Sciences, Block JD, Sector III, Salt Lake, Kolkata 700 106, India}

\author{A. S. Majumdar}
\email{archan@bose.res.in}
\affiliation{S. N. Bose National Centre for Basic Sciences, Block JD, Sector III, Salt Lake, Kolkata 700 106, India}

\begin{abstract}
Can the entanglement of an arbitrary many-body state be quantified without reconstructing the state? We answer the question affirmatively by showing how the measure of logarithmic negativity can be estimated using a finite number of even partial-transpose (PT) moments of a many-body state. 
Considering the squared PT spectrum as a probability measure, we solve
exactly the task of a truncated Hausdorff moment problem with Gauss and Gauss-Radau quadrature. We obtain a hierarchy of lower bounds  on the logarithmic negativity, which improves with increasing order of moments.
Convergence to the exact value of the measure can be achieved for every state of finite dimension.  We illustrate the bounds for ground states of three different interacting spin models, as well as for a thermal state. An error analysis demonstrates the efficacy of our scheme towards
handling the limitation of finite statistics encountered in actual experiments with randomized measurements. Moreover, an example of density matrix renormalization group based numerical estimation is provided for handling system sizes beyond the scope of exact diagonalization.
\end{abstract}

\maketitle

\textit{Introduction.---}Entanglement is both a resource for quantum information processing and a window into collective quantum phenomena~\cite{HorodeckiRMP2009, AmicoFazioOsterlohVedral2008}. It enables teleportation of an unknown state \cite{Bennett1993Teleporting}, underlies quantum cryptography \cite{Ekert1991}, and improves precision beyond the shot-noise limit \cite{GiovannettiLloydMaccone2004,GiovannettiLloydMaccone2011}. In many-body systems it reveals quantum criticality through universal scaling \cite{VidalLatorreRicoKitaev2003,CalabreseCardy2004} and exposes topological order that no local observable can see \cite{KitaevPreskill2006,LevinWen2006}. The states dealt with in experiments and simulations are mixed, because of noise, finite temperature, or degrees of freedom that cannot be accessed. For such states, detecting entanglement is only the first step. What one needs to know is how much of it survives \cite{HorodeckiRMP2009,AmicoFazioOsterlohVedral2008,VidalTarrach1999,Vidal2000,TerhalHorodecki2000}.

The logarithmic negativity is one of the few mixed-state entanglement monotones that can be computed directly from a density matrix \cite{Peres1996,Horodecki1996,Zyczkowski1998,VidalWerner2002,Plenio2005}. This advantage comes at a price. Its exact value requires the full spectrum of the partial transpose (PT), and for an unknown state this requires full tomography~\cite{d2003quantum, odonnell2016efficient, christandl2012reliable, haah2017sample}, whose cost grows rapidly with system size \cite{NielsenChuang2011}. Entanglement is among the most informative quantities of a many-body state, while the data needed to compute it are among the least accessible. One would like to quantify entanglement from partial information, without reconstructing the state.

Moments of the partial transpose (PT) offer such information. They are multicopy observables that can be estimated with swap-type circuits and randomized measurements \cite{Carteret2005,GrayBanchiBayatBose2018,Elben2019Statistical,Huang2020Predicting,Zhou2020,Elben2020}, the same tools that enable experimental access to R\'enyi entropies \cite{Islam2015,Brydges2019,Rath2023}. Moments serve adequately as witnesses of entanglement and of other quantum resources \cite{HorodeckiEkert2002,GuhneToth2009,Aggarwal2024,PhysRevA.109.022247,Mallick2025,Mukherjee2025,PhysRevA.111.032406,lkcb-nqqb,x56c-2yp1,chakrabarty2026operationalcertificationnonclassicalityarbitrary,chakrabarty2026detectionquantumimaginarityusing,chakrabarty2026quantifyingmargenauhillnonclassicality}. For PT moments the detection problem is now well understood. The third moment already gives a strong test \cite{Elben2020}, a full set of moment conditions is known \cite{Neven2021,Miller_2026}, and the optimal entanglement test for a given set of PT moments has been constructed \cite{YuImaiGuhne2021}. The problem of quantification is much less studied. Quantitative witnesses bound entanglement measures from expectation values \cite{EisertBrandaoAudenaert2007}, and it was shown recently that the two moments of order two and four yield a lower bound on the logarithmic negativity \cite{TarabungaHaug2025}. This leaves a basic yet important question open. {\textit{Given a finite set of PT moments, how much logarithmic negativity can one certify, and which bound is the strongest that these numbers can support?}}

Here we answer this question for the even PT moments. We take the squared PT spectrum as a classical probability measure. The even moments then define a truncated Hausdorff moment problem, and the negativity becomes the integral of a fixed function against the unknown measure \cite{ShohatTamarkin1943,AkhiezerBook,KarlinStudden1966,KreinNudelman1977,CurtoFialkow1991}. We solve the resulting optimization exactly with Gauss and Gauss-Radau quadrature \cite{GautschiBook}. The outcome is a hierarchy of lower bounds on the logarithmic negativity. Each member is the best bound that its moment data can give within the moment problem. The hierarchy improves strictly until the distinct magnitudes of the PT spectrum are resolved, becomes exact at that point, and converges to the logarithmic negativity. Its first member is the bound of Ref.~\cite{TarabungaHaug2025}, which is therefore optimal for two moments. 

Through our formalism, we verify the above hierarchy on various categories
of quantum states. We first establish two-qubit benchmarks considering a
non-maximally entangled pure state and also the mixed Werner state. More
interesting examples of many-body entangled states are then furnished 
through several spin-chain models, with inclusion of the effect of thermal noise, as well. Large many-body states are not amenable to methods using exact diagonalization, and for such cases we provide an example of numerical estimation using the density matrix renormalization group (DMRG). Moreover, the suitability of our moments-based method for quantifying many-body
entanglement is further revealed through an error analysis required for handling the
effect of finite-size statistics in actual experiments involving randomized measurements.   

\textit{Setup.---}Let $\rho$ be a state on $\mathcal H_A\otimes\mathcal H_B$, let $\rho^\Gamma$ be its partial transpose on $B$, and let $\{\lambda_i\}$ be the real eigenvalues of $\rho^\Gamma$. The logarithmic negativity and the PT moments are
\begin{align}
\mathcal E(\rho)&=\ln\|\rho^\Gamma\|_1=\ln\sum_i|\lambda_i|, \label{eq:E_def}\\
p_k(\rho)&=\mathrm{tr}\big[(\rho^\Gamma)^k\big]=\sum_i\lambda_i^k, \label{eq:pk_def}
\end{align}
with $k$ a positive integer. Normalization gives $p_1=1$, and $p_2$ is the purity of $\rho$. We use $\mathcal E$ for the logarithmic negativity throughout. The even moments depend only on the magnitudes $|\lambda_i|$, which is also all that $\mathcal E$ depends on. We therefore introduce the normalized squared eigenvalues and their power sums,
\begin{align}
q_i=\frac{\lambda_i^2}{p_2}, \qquad \mu_k=\sum_i q_i^k=\frac{p_{2k}}{p_2^k}. \label{eq:qi}
\end{align}
The $q_i$ are non-negative and sum to one, so they form a probability distribution, and $\mu_1=1$. We call $\mu_k$ the reduced moments. Let $\nu$ be the measure on $[0,1]$ that places weight $q_i$ at the point $q_i$, for every nonzero $q_i$. With the function $h(x)=x^{-1/2}$ one finds
\begin{align}
\mu_k=\int_0^1 x^{k-1}\,d\nu, \qquad S:=\sum_i q_i^{1/2}=\int_0^1 h\,d\nu, \label{eq:nu_def}
\end{align}
and the logarithmic negativity takes the form
\begin{align}
\mathcal E(\rho)=\tfrac12\ln p_2+\ln S. \label{eq:E_via_S}
\end{align}
The first term is fixed by one measured moment. All the difficulty sits in $S$, the integral of a convex function that diverges at the origin against a measure of which only a few moments are known.

Suppose that $p_2,p_4,\dots,p_K$ are given for an even $K\ge4$. This fixes the $r=K/2$ reduced moments $\mu_1,\dots,\mu_r$. Many probability measures on $[0,1]$ share these moments, and $\nu$ is one of them. The smallest value of the integral of $h$ over all of them,
\begin{align}
S_K^-:=\inf_\sigma\Big\{\int_0^1 h\,d\sigma \;:\; \int_0^1 x^{k-1}d\sigma=\mu_k,\ 1\le k\le r\Big\}, \label{eq:SK_def}
\end{align}
where $\sigma$ runs over probability measures on $[0,1]$, cannot exceed $S$. This defines the $K$-moment quantifier
\begin{align}
\mathcal E_K(\rho):=\tfrac12\ln p_2+\ln S_K^-\ \le\ \mathcal E(\rho). \label{eq:EK_def}
\end{align}
Equation~\eqref{eq:EK_def} is a valid bound by construction. It does not yet tell us the value of $S_K^-$, nor which measure attains it.

\textit{Optimal moment hierarchy.---}The infimum in Eq.~\eqref{eq:SK_def} can be found exactly. Every derivative of $h$ has a fixed sign on $(0,1]$, positive at even order and negative at odd order. This allows one to build, by Hermite interpolation, a polynomial that lies below $h$ and touches it at the nodes of a quadrature rule. The rule is then both an admissible measure and a certificate that no admissible measure does better. The following theorem states the result, with $s$ the number of distinct nonzero values among the $|\lambda_i|$.

\begin{theorem}[Optimal $K$-moment quantifier]
\label{thm:main}
Let $K\ge4$ be even, $r=K/2$ and $\ell=\lfloor K/4\rfloor$. The infimum $S_K^-$ in Eq.~\eqref{eq:SK_def} is attained by a unique measure $\nu_K^\star$. If $s\le\ell$, then $\nu_K^\star=\nu$. If $s>\ell$ and $r$ is even, $\nu_K^\star$ is the $\ell$-point Gauss quadrature rule of the moments $\mu_1,\dots,\mu_r$, with all nodes in $(0,1)$. If $s>\ell$ and $r$ is odd, $\nu_K^\star$ is the Gauss-Radau rule with $\ell$ nodes in $(0,1)$ and one node fixed at $x=1$. In every case, with nodes $x_j$ and weights $w_j$,
\begin{align}
\mathcal E_K(\rho)=\tfrac12\ln p_2+\ln\sum_j w_j\,x_j^{-1/2}. \label{eq:EK_gauss}
\end{align}
\end{theorem}

The idea of the proof is given in Appendix~A of the End Matter, and the full proof in  \hyperref[sm:proof]{Sec.~S2 of the Supplemental Material (SM)}~\cite{SM}. Since $\nu_K^\star$ reproduces every given moment, no function of $\mu_1,\dots,\mu_r$ that is a valid bound for all measures can exceed $S_K^-$. In this sense $\mathcal E_K$ is the strongest bound that the even moments support.

The theorem has a simple reading. Among all spectra compatible with the data, the least entangled one has the fewest levels, namely $\ell$ magnitudes placed so that the measured moments are reproduced. For $K=4$ this is Jensen's inequality for the convex function $h$, which replaces the spectrum by a single level at its mean. The higher orders refine Jensen's inequality with the help of further moments.

The scope of this statement should be made clear. The moment problem keeps the even moments only. It uses neither the odd moments nor the normalization $p_1=1$. The normalization alone gives $\sum_i|\lambda_i|\ge1$, that is $\mathcal E\ge0$, so the value actually certified is
\begin{align}
\mathcal E_K^+(\rho)=\max\{\mathcal E_K(\rho),\,0\}. \label{eq:EK_plus}
\end{align}
The role of the odd moments is discussed in  \hyperref[sm:proof]{SM~S2}~\cite{SM}.

For $K=4$ the rule is a single node at $x=\mu_2$ with unit weight, and Eq.~\eqref{eq:EK_gauss} becomes
\begin{align}
\mathcal E_4(\rho)=\frac12\ln\frac{p_2^3}{p_4}. \label{eq:E4}
\end{align}
This is the bound of Ref.~\cite{TarabungaHaug2025}. Theorem~\ref{thm:main} adds that nothing better can be obtained from $p_2$ and $p_4$ within the moment problem. For $K=6$ the rule has one free node $a$ and the fixed node $x=1$, which gives
\begin{align}
a&=\frac{\mu_2-\mu_3}{1-\mu_2}, \qquad w_a=\frac{(1-\mu_2)^2}{1-2\mu_2+\mu_3}, \label{eq:E6_node}\\
\mathcal E_6(\rho)&=\tfrac12\ln p_2+\ln\!\big[w_a\,a^{-1/2}+1-w_a\big]. \label{eq:E6}
\end{align}
For $K=8$ the rule has two free nodes $x_1>x_2$, the roots of $x^2-\alpha x+\beta$ with
\begin{align}
\alpha=\frac{\mu_4-\mu_2\mu_3}{\mu_3-\mu_2^2}, \qquad \beta=\alpha\mu_2-\mu_3, \label{eq:E8_nodes}
\end{align}
and the weights are $w_1=(\mu_2-x_2)/(x_1-x_2)$ and $w_2=1-w_1$. The quantifier $\mathcal E_8$ then follows from Eq.~\eqref{eq:EK_gauss}. The derivation of these forms, and the proof that all weights are positive, are given in  \hyperref[sm:small_K]{SM~S3}~\cite{SM}. For larger $K$ the nodes and weights follow from one small eigenvalue problem (\hyperref[sm:higherK]{SM~S7}~\cite{SM}).

Each new moment adds one linear constraint to Eq.~\eqref{eq:SK_def}. The admissible set can only shrink, so
\begin{align}
\mathcal E_4(\rho)\le\mathcal E_6(\rho)\le\mathcal E_8(\rho)\le\cdots\le\mathcal E(\rho). \label{eq:monotone}
\end{align}
One may ask whether a further moment always helps. The next result shows that it does, as long as something is left to gain.

\begin{proposition}[Strict improvement]
\label{prop:strict}
$\mathcal E_K(\rho)<\mathcal E_{K+2}(\rho)$ if and only if $\mathcal E_K(\rho)<\mathcal E(\rho)$.
\end{proposition}

The proof is given in  \hyperref[sm:saturation]{SM~S4}~\cite{SM}. The hierarchy thus rises strictly until it reaches $\mathcal E$ and is constant afterwards. The order at which this happens can be read off from the spectrum.

\begin{theorem}[Saturation]
\label{thm:saturation}
$\mathcal E_K(\rho)=\mathcal E(\rho)$ if and only if the PT spectrum has at most $\lfloor K/4\rfloor$ distinct nonzero magnitudes $|\lambda_i|$.
\end{theorem}

The proof follows from the uniqueness of $\nu_K^\star$ and is given in  \hyperref[sm:saturation]{SM~S4}~\cite{SM}. A spectrum with $s$ magnitudes is therefore resolved exactly at the order
\begin{align}
K_{\rm sat}=4s. \label{eq:Ksat}
\end{align}
This is the central statement on when the hierarchy converges. A state of finite dimension has finitely many magnitudes, so the hierarchy always terminates, and from $K_{\rm sat}$ on it returns $\mathcal E$ itself and not a bound. The fixed node of the Gauss-Radau rule does not resolve a new magnitude, so $\mathcal E_6$ improves on $\mathcal E_4$ without saturating on any state on which $\mathcal E_4$ does not. An important family is resolved already at the lowest order.

\begin{corollary}[Stabilizer states]
\label{cor:stabilizer}
For every entangled pure stabilizer state, $\mathcal E_4(\rho)=\mathcal E_6(\rho)=\cdots=\mathcal E(\rho)$.
\end{corollary}

The reason is that such a state has a flat Schmidt spectrum \cite{HeinEisertBriegel2004}, so all nonzero $|\lambda_i|$ are equal (\hyperref[sm:saturation]{SM~S4}~\cite{SM}). A generic many-body state lies at the opposite extreme. Its number of magnitudes $s$ can be in the hundreds, so that $K_{\rm sat}$ is out of reach. The relevant question is then how the bound behaves before saturation. By Proposition~\ref{prop:strict} it increases at every step, and the next theorem shows that it approaches $\mathcal E$ whatever the value of $s$.

\begin{theorem}[Convergence]
\label{thm:convergence}
For every state $\rho$, the quantifiers $\mathcal E_K(\rho)$ converge to $\mathcal E(\rho)$ as $K\to\infty$.
\end{theorem}

The proof is given in  \hyperref[sm:convergence]{SM~S5}~\cite{SM}. It uses the fact that a measure on a compact interval is fixed by all its moments, so the admissible sets shrink to $\nu$ alone. The argument does not refer to the dimension of the state. The price of a higher order is experimental. Estimating $p_K$ needs $K$ copies in a swap-type circuit, or a number of randomized measurements that grows with $K$ and exponentially with the size of the subsystem \cite{Elben2020,Zhou2020,Neven2021}. The orders $K\le8$, given here in closed form, are the ones within present reach.

\begin{figure}[t]
\centering
\includegraphics[width=\linewidth]{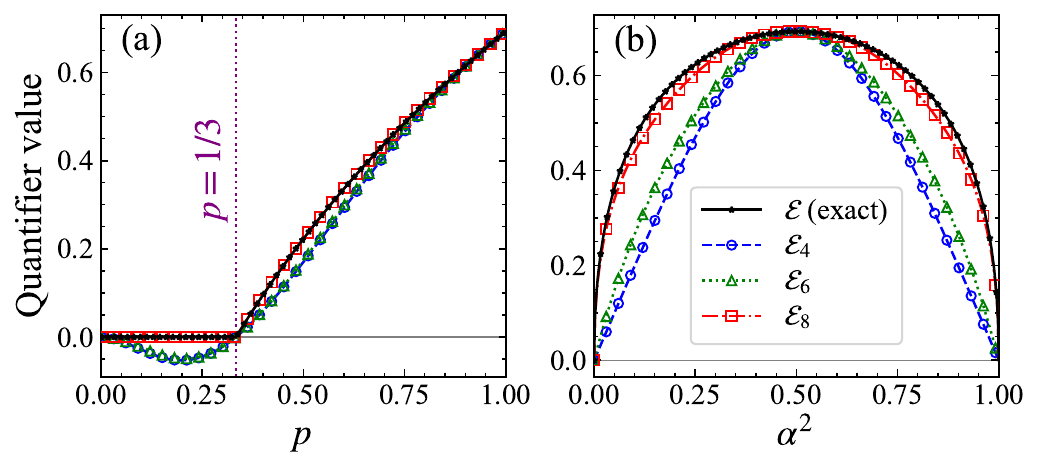}
\caption{Logarithmic negativity $\mathcal E$ and the quantifiers $\mathcal E_4$, $\mathcal E_6$, $\mathcal E_8$ for (a) the Werner state of Eq.~\eqref{eq:werner} as a function of $p$, and (b) the pure state of Eq.~\eqref{eq:pure} as a function of $\alpha^2$. In (a) the PT spectrum has two magnitudes and $\mathcal E_8$ coincides with $\mathcal E$, as predicted by Theorem~\ref{thm:saturation}. The vertical line marks the separability threshold $p=1/3$. In (b) the spectrum has three magnitudes, so $\mathcal E_8$ stays below $\mathcal E$ except at $\alpha^2=0,\,1/2,\,1$, and saturation needs $K=12$.}
\label{fig:qubit}
\end{figure}

\begin{figure*}[t]
\centering
\includegraphics[width=0.8\linewidth]{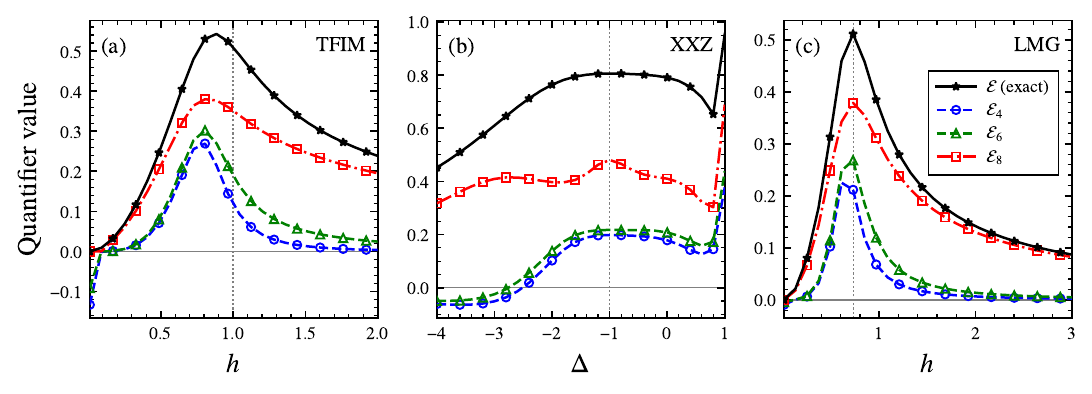}
\caption{Logarithmic negativity $\mathcal E$ and the quantifiers $\mathcal E_4$, $\mathcal E_6$, $\mathcal E_8$ for the ground state of (a) the transverse-field Ising model of Eq.~\eqref{eq:H_tfim} as a function of the field $h$, (b) the XXZ chain of Eq.~\eqref{eq:H_xxz} as a function of the anisotropy $\Delta$, and (c) the LMG model of Eq.~\eqref{eq:H_lmg} as a function of $h$. In all panels $n=12$, with two blocks of five sites in the middle kept and one site at each end traced out. The dotted lines mark the critical point $h=1$ in (a), the boundaries $\Delta=\pm1$ of the critical phase in (b), and the finite-size entanglement peak $h\approx0.73$ in (c). The ordering of Eq.~\eqref{eq:monotone} holds at every point, and $\mathcal E_8$ recovers a large part of the gap left by $\mathcal E_4$ and $\mathcal E_6$.}
\label{fig:manybody}
\end{figure*}

\textit{Two-qubit benchmarks.---}Theorem~\ref{thm:saturation} predicts the order at which the bound becomes exact. We test this on two families with a known PT spectrum. The first is the Werner state \cite{Werner1989},
\begin{align}
\rho_W(p)=p\,|\Phi^+\rangle\langle\Phi^+|+(1-p)\,\frac{\mathbb I_4}{4}, \label{eq:werner}
\end{align}
whose partial transpose has the eigenvalue $(1+p)/4$ three times and $(1-3p)/4$ once. There are two magnitudes, which merge at $p=0$ and $p=1$, and one of them vanishes at the separability threshold $p=1/3$. Saturation is therefore predicted at $K=8$. The second family is the pure state
\begin{align}
|\psi(\alpha)\rangle=\alpha|00\rangle+\sqrt{1-\alpha^2}\,|11\rangle, \label{eq:pure}
\end{align}
with PT eigenvalues $\alpha^2$, $1-\alpha^2$ and $\pm\alpha\sqrt{1-\alpha^2}$. These are three magnitudes, so saturation needs $K=12$, except at $\alpha^2=0,\,1/2,\,1$ where they merge. Fig.~\ref{fig:qubit} and Fig.~\ref{fig:Kconv_qubit} confirms both predictions. For the Werner state at $p=1/2$ we find $\mathcal E_4=0.170$, $\mathcal E_6=0.174$ and $\mathcal E_8=\mathcal E=0.223$. The step to $\mathcal E_6$ is small and $\mathcal E_8$ closes the gap, as expected, since the sixth moment adds no new level and the eighth adds one. For $p<1/3$ the state is separable and $\mathcal E_4$, $\mathcal E_6$ are negative, so that Eq.~\eqref{eq:EK_plus} returns zero.

\textit{Many-body entanglement.---}The natural setting for a moment-based quantifier is an extended system, where tomography is out of reach and the growth of entanglement signals a quantum phase transition \cite{AmicoFazioOsterlohVedral2008,SchuchWolfVerstraeteCirac2008,Cirac2021mps}. We consider the ground states of three spin-$\tfrac12$ models on $n=12$ sites, the transverse-field Ising model (TFIM), the XXZ chain and the Lipkin-Meshkov-Glick (LMG) model \cite{Lipkin1965,LatorreOrusRicoVidal2005} respectively, with the Hamiltonians,
\begin{align}
H_{\mathrm{TFIM}}&=-\sum_{k}\sigma^x_k\sigma^x_{k+1}-h\sum_k\sigma^z_k, \label{eq:H_tfim}\\
H_{\mathrm{XXZ}}&=-\sum_{k}\big(\sigma^x_k\sigma^x_{k+1}+\sigma^y_k\sigma^y_{k+1}+\Delta\,\sigma^z_k\sigma^z_{k+1}\big), \label{eq:H_xxz}\\
H_{\mathrm{LMG}}&=-\frac{1}{n-1}\sum_{i<j}\sigma^x_i\sigma^x_j-h\sum_k\sigma^z_k. \label{eq:H_lmg}
\end{align}
They cover an isolated critical point at $h=1$ (TFIM), a gapless phase for $-1\le\Delta<1$ (XXZ), and an infinite-range interaction without any notion of distance (LMG), whose entanglement peaks at $h\approx0.73$ for this size. Two adjacent blocks $A$ and $B$ of five sites each in the middle of the chain are kept, and the two end sites $C$ are traced out. The state of $AB$ is then mixed, and its partial transpose is diagonalized exactly. Details are given in  \hyperref[sm:manybody]{SM~S6}~\cite{SM}.

\begin{figure}[t]
\centering
\includegraphics[width=0.9\linewidth]{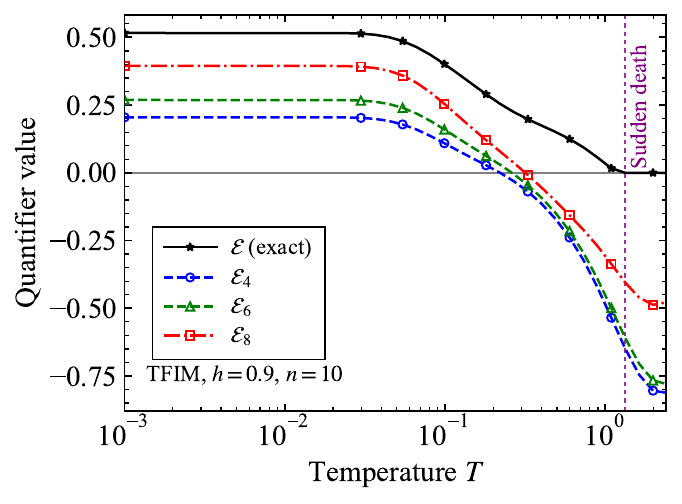}
\caption{Logarithmic negativity $\mathcal E$ and the quantifiers $\mathcal E_4$, $\mathcal E_6$, $\mathcal E_8$ for the Gibbs state of the transverse-field Ising model at $h=0.9$ on $n=10$ sites as a function of the temperature $T$. Two blocks of four sites are kept and the two end sites are traced out. The negativity vanishes at $T\approx1.21$ (dashed line). Each $\mathcal E_K$ crosses zero at a lower temperature and is negative beyond it, so that the certified value $\mathcal E_K^+$ of Eq.~\eqref{eq:EK_plus} is zero there.}
\label{fig:thermal}
\end{figure}

Figure~\ref{fig:manybody} shows the result. The ordering of Eq.~\eqref{eq:monotone} holds at every scanned point of the three models. In the TFIM all quantifiers vanish for $h\to0$, where the ground state is a cat state and the reduced state of $AB$ is separable. They peak just below $h=1$ and decay as the chain polarizes. Around the peak the gap to $\mathcal E$ is widest. The reason is physical. Close to the critical point the correlation length exceeds the size of the blocks, many Schmidt levels of comparable weight contribute, and the PT spectrum inherits their spread. At $h=0.85$ we find $\mathcal E=0.542$ while $\mathcal E_4=0.245$, $\mathcal E_6=0.290$ and $\mathcal E_8=0.380$. Two further moments thus recover close to half of what $\mathcal E_4$ misses, and most of this gain comes from $p_8$, as for the Werner state. This is what Theorem~\ref{thm:saturation} leads one to expect. The spectrum has several hundred distinct magnitudes, so Eq.~\eqref{eq:Ksat} places exact saturation beyond $K=10^3$. In the paramagnet the eighth moment matters even more. At $h=2$ the bound $\mathcal E_4$ certifies almost nothing, while $\mathcal E_8$ recovers about $80\%$ of $\mathcal E$. 

The XXZ chain shows the same effect over a whole interval. The gap stays wide across the critical region. In the antiferromagnet, for $\Delta\lesssim-2.5$, the two lowest orders are negative and certify no entanglement, whereas $\mathcal E_8$ still recovers more than half of $\mathcal E$. The last point, $\Delta=1$, is the isotropic ferromagnet, where the ground state is degenerate and the values depend on the state selected. 

The LMG model has no notion of distance, and its ground state lies in the permutation-symmetric subspace of dimension $n+1$. Yet the picture is the same, with the gap widest at the entanglement peak. The performance of the hierarchy is therefore set by the spread of the PT spectrum and not by the range of the interaction. In all three models the step to $p_8$ is the one that pays most.

The size $n=12$ is a limit of exact diagonalization, not of the method. With ground states from the density-matrix renormalization group (DMRG) \cite{White1992,Schollwock2011} we follow the same ten-site window inside chains of up to $n=32$ sites for the TFIM and $n=36$ for the XXZ chain. The negativity of the window decreases as the traced-out environment grows, and the ordering of Eq.~\eqref{eq:monotone} holds at every size (Appendix~D of the End Matter).

The states above are mixed because a part of a pure state is discarded. A thermal state is mixed for a different reason. Figure~\ref{fig:thermal} shows the Gibbs state of the TFIM on $n=10$ sites at $h=0.9$ as a function of the temperature $T$, again with the two end sites traced out. The negativity decreases with $T$ and vanishes at a finite temperature $T\approx1.21$, in units of the coupling. Such a threshold temperature is a known feature of thermal entanglement \cite{ArnesenBoseVedral2001,HartCastelnovo2018} and is often called sudden death \cite{YuEberly2009}. Every $\mathcal E_K$ crosses zero before $\mathcal E$ does, and none is positive above the threshold. Beyond it the raw $\mathcal E_K$ become negative, and Eq.~\eqref{eq:EK_plus} gives the correct value zero.

How fast does the hierarchy converge for such spectra? We have followed $\mathcal E_K$ up to $K\approx60$ for one state of each model (Appendix~B of the End Matter, Figs.~\ref{fig:Kconv_qubit} and \ref{fig:Kconv_manybody}). For the two-qubit states the bound becomes exact at the predicted orders $K_{\rm sat}=8$ and $12$. For the many-body states the gap shrinks with every moment, by more than three orders of magnitude for the LMG state, which has $s=36$ and $K_{\rm sat}=144$, and more slowly for the two chains, for which $K_{\rm sat}$ exceeds $10^3$. Before saturation the remaining gap is controlled by the smallest magnitudes, where the function $h$ is large and steep, so a long tail of small eigenvalues, typical of critical states, slows the convergence down.

\textit{Conclusions.---} Moments of the partial transpose,  signalling entanglement, can be measured on large systems. In the present work we have determined what they are worth quantitatively. Taken as a truncated moment problem, the even PT moments fix a best lower bound on the logarithmic negativity at every order, and the best bound is given by a Gauss or Gauss-Radau quadrature rule. The bounds form a hierarchy with a complete convergence picture. Every further moment improves the bound strictly, the bound is exact for $K_{\rm sat}=4s$,  and it converges to the negativity in general. 

Our results signify what a measurement of PT moments delivers. Data taken to detect entanglement can be turned into a quantifying number with a guarantee, at no further experimental cost, and the saturation criterion tells in advance for which states a few moments are enough. Stabilizer states and states with few spectral levels are quantified exactly. For critical many-body states the low orders certify a sizable part of the entanglement, and the remaining gap is itself a signature of a broad PT spectrum. Since only moments enter, the bounds apply equally to experiments with randomized measurements and to numerical methods such as DMRG (Appendix~D of the End Matter).

Several open avenues for further analysis follow naturally from this work. The first concerns the information that the present bound leaves aside. The odd moments carry the signs of the PT eigenvalues, and the third moment is the cheapest one to measure.  The second issue is statistical. A lower bound at a stated confidence level requires the variances of a concrete protocol and an optimization over the admissible moment sequences (Appendix~C of the End Matter). The third is the speed of convergence before saturation. We find that it is governed by the tail of small eigenvalues, but a general estimate is desirable. Finally, the construction rests only on the alternating signs of the derivatives of one function. It should therefore extend to realignment moments \cite{Rudolph2004,ChenWu2003,TarabungaHaug2025} and to other spectral quantities, for instance the von Neumann entropy in terms of the moments of the density matrix. Semidefinite methods \cite{BertsimasPopescu2005,Lasserre2009} offer a route when further constraints are added. Optimal bounds from a few moments could then become a general tool for quantifying quantum resources in systems that are too large for tomography.

\vspace{10mm}
\begin{center}
\textbf{ END MATTER}
\end{center}

\section*{Appendix A: Outline of the proof of Theorem~\ref{thm:main}}\label{app:A}

The derivatives of $h(x)=x^{-1/2}$ are
\begin{align}
h^{(m)}(x)=(-1)^m\binom{2m}{m}\frac{m!}{4^m}\,x^{-m-1/2}, \label{eq:h_derivatives}
\end{align}
positive for even $m$ and negative for odd $m$ on $(0,1]$. Take first $r=2\ell$ and let $x_1,\dots,x_\ell$ be the Gauss nodes. Let $P$ be the polynomial of degree at most $2\ell-1$ that agrees with $h$ in value and slope at every node. The Hermite remainder formula gives, for some $\xi$ in $(0,1]$,
\begin{align}
h(x)-P(x)=\frac{h^{(2\ell)}(\xi)}{(2\ell)!}\prod_{j=1}^{\ell}(x-x_j)^2\ \ge0. \label{eq:em_gauss}
\end{align}
So $P\le h$, and the integral of $P$ against any admissible measure is fixed by the given moments. This is a lower bound on $S_K^-$. The Gauss rule is admissible and sits on the points where $P=h$, so it attains the bound. For $r=2\ell+1$ one more condition is available, and $P$ is also made to agree with $h$ at a fixed point $c$. The remainder becomes
\begin{align}
h(x)-P(x)=\frac{h^{(2\ell+1)}(\xi)}{(2\ell+1)!}\,(x-c)\prod_{j=1}^{\ell}(x-x_j)^2. \label{eq:em_radau}
\end{align}
Since $h^{(2\ell+1)}<0$, this is non-negative on $[0,1]$ only for $c=1$. The choice $c=0$ would give an upper bound. The fixed node at $x=1$ is thus forced by the sign and not chosen for convenience. The complete argument, including uniqueness and the case $s\le\ell$, is given in  \hyperref[sm:proof]{SM S2}~\cite{SM}.

\section*{Appendix B: Convergence of the hierarchy}\label{app:B}

Theorems~\ref{thm:saturation} and \ref{thm:convergence} state when and how the bound reaches $\mathcal E$. We test both statements with the general quadrature rule of Theorem~\ref{thm:main}, evaluated in arithmetic with $100$ decimal digits (\hyperref[sm:higherK]{SM S7}~\cite{SM}).

Figure~\ref{fig:Kconv_qubit} shows $\mathcal E_K$ for the Werner state at $p=1/2$ and for the pure state of Eq.~\eqref{eq:pure} with $\alpha=1/2$. The quantifier increases at every step and reaches $\mathcal E$ at $K=8$ and $K=12$, respectively. From there on it is constant. These are the orders $K_{\rm sat}=4s$ of Eq.~\eqref{eq:Ksat} for spectra with $s=2$ and $s=3$ magnitudes, so the hierarchy terminates exactly where Theorem~\ref{thm:saturation} says it should.

\begin{figure}[t]
\centering
\includegraphics[width=0.77\linewidth]{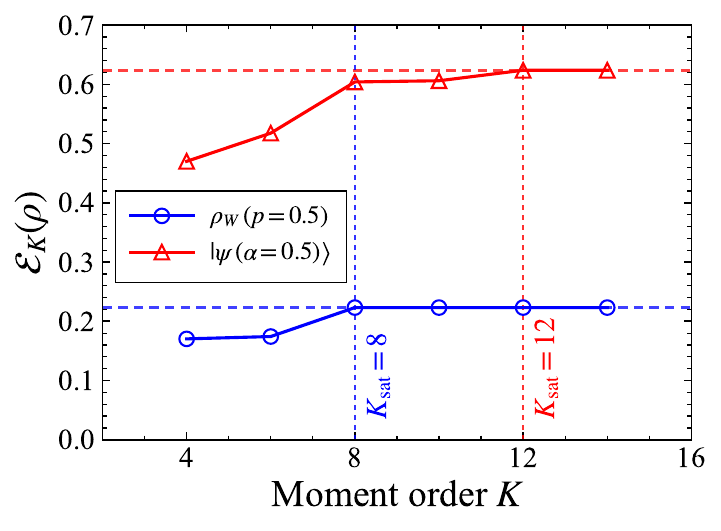}
\caption{Quantifier $\mathcal E_K$ as a function of the order $K$ for the Werner state at $p=1/2$ and the pure two-qubit state with $\alpha=1/2$. The horizontal lines mark $\mathcal E$. The vertical lines mark the orders $K_{\rm sat}=8$ and $K_{\rm sat}=12$ at which the bound becomes exact, in agreement with Theorem~\ref{thm:saturation}.}
\label{fig:Kconv_qubit}
\end{figure}

Figure~\ref{fig:Kconv_manybody} shows the same scan for the three many-body ground states, the TFIM at $h=0.85$, the XXZ chain at $\Delta=-0.5$ and the LMG model at $h=0.73$. Here the bound rises steadily and the gap falls by orders of magnitude, but no saturation occurs in the range shown. This is expected. The LMG state has $s=36$ distinct magnitudes, a consequence of its permutation symmetry, which gives $K_{\rm sat}=144$. Its gap closes fastest and is $10^{-4}$ at $K=60$. The TFIM and XXZ states have several hundred magnitudes each, with $K_{\rm sat}$ above $10^3$. Their gaps close more slowly, to $4\times10^{-3}$ at $K=66$ for the TFIM state and $2\times10^{-2}$ at $K=80$ for the XXZ state. The hierarchy thus converges for every state, quickly when the PT spectrum has few levels and slowly when it has a long tail of small eigenvalues. Orders of this size cannot be measured at present, but they are accessible  numerically, where the moments of a state can be computed to any order.

\begin{figure}[t]
\centering
\includegraphics[width=\linewidth]{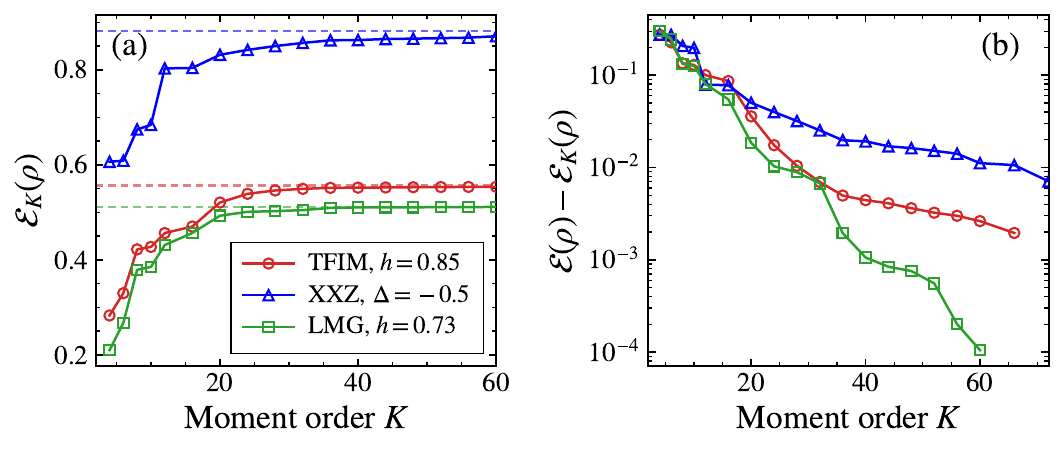}
\caption{Convergence of $\mathcal E_K$ for the ground states of the TFIM at $h=0.85$, the XXZ chain at $\Delta=-0.5$ and the LMG model at $h=0.73$, all for $n=12$. (a) $\mathcal E_K$ as a function of $K$. The dashed lines mark $\mathcal E$. (b) The gap $\mathcal E-\mathcal E_K$ is on a logarithmic scale. The gap decreases monotonically, as required by Proposition~\ref{prop:strict}, and fastest for the LMG state, which has the fewest distinct magnitudes.}
\label{fig:Kconv_manybody}
\end{figure}

\section*{Appendix C: Finite statistics}\label{app:C}

An experiment returns estimates $\hat p_{2k}$ and not the exact moments. Since $\mathcal E_K$ is nonlinear in the moments, the reduced moments $\hat\mu_k=\hat p_{2k}/\hat p_2^{\,k}$ need not belong to any probability measure on $[0,1]$. For $K=6$ the admissible region is $0<\mu_2<1$ and $\mu_2^2<\mu_3<\mu_2$ (\hyperref[sm:small_K]{SM S3}~\cite{SM}). Outside it, Eq.~\eqref{eq:E6} is not defined. This matters for near-term devices, where the number of copies is a limited resource \cite{Preskill2018}. To see how often it occurs we use a simple model. Each of $p_2,p_4,p_6$ receives independent Gaussian noise of relative width $1/\sqrt{N_s}$, with $N_s$ the number of shots. This is a generic scaling and not the variance of a specific protocol \cite{Elben2020,Zhou2020}. We take the Werner state at $p=1/2$ and the TFIM ground state at $h=0.85$ of Fig.~\ref{fig:manybody}, with $8000$ trials for every $N_s$. As estimator we use $\mathcal E_6$ when the data are admissible and fall back to $\mathcal E_4$ otherwise.

Figure~\ref{fig:noise_robustness} shows the outcome. For the TFIM state a fifth of the trials are inadmissible at $N_s=10^2$ and fewer than one percent at $N_s=10^3$. The Werner state is harder, because its spectrum is close to flat and the admissible window for $\mu_3$ is narrow. Almost half of the trials fail at $N_s=10^2$, and the fraction drops below one percent only between $N_s=3\times10^4$ and $10^5$. The root-mean-square error with respect to $\mathcal E$ decreases with $N_s$ and then levels off at the gap $\mathcal E-\mathcal E_6$ of the exact moments, which is $0.049$ for the Werner state and $0.252$ for the TFIM state. Beyond this point more shots do not help and only a larger $K$ does. The estimator is a point estimate. A lower bound at a given confidence level would require minimizing $\mathcal E_K$ over a confidence region of the moments, intersected with the admissible set. We leave this for future work.

\begin{figure}[t]
\centering
\includegraphics[width=\linewidth]{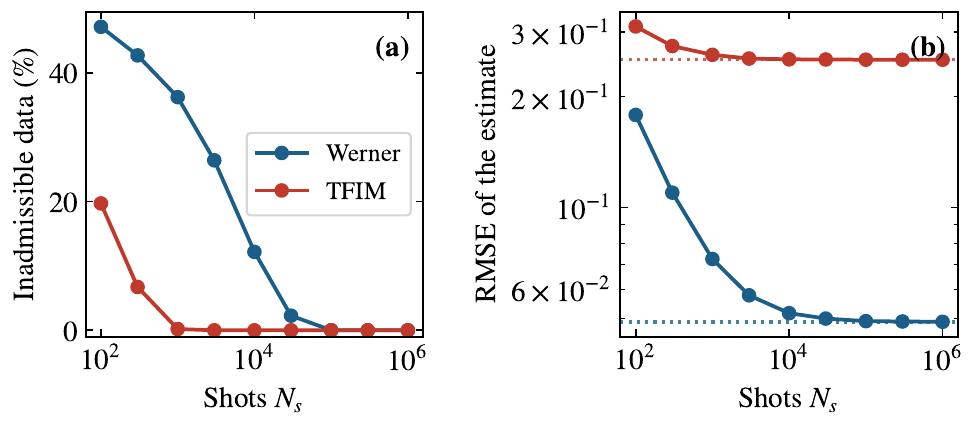}
\caption{Finite statistics for $K=6$, for the Werner state at $p=1/2$ and the TFIM ground state at $h=0.85$ ($n=12$), with independent Gaussian noise of relative width $1/\sqrt{N_s}$ on $p_2$, $p_4$, $p_6$ and $8000$ trials per point. (a) Fraction of trials in which the estimated reduced moments are not admissible. (b) Root-mean-square error of the estimate with respect to $\mathcal E$. The estimate is $\mathcal E_6$ for admissible data and $\mathcal E_4$ otherwise. The dotted lines mark the gap $\mathcal E-\mathcal E_6$ of the exact moments.}
\label{fig:noise_robustness}
\end{figure}

\begin{figure}[h]
\centering
\includegraphics[width=\linewidth]{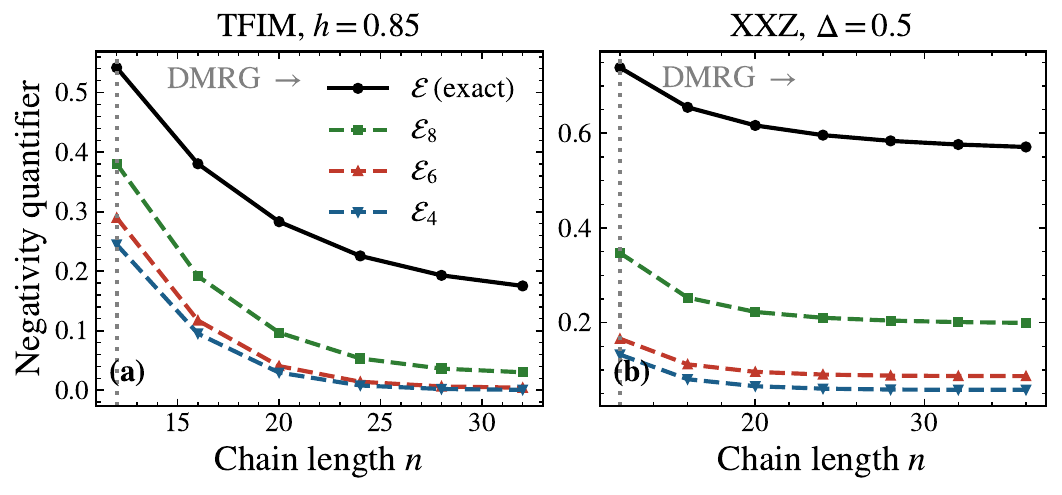}
\caption{Logarithmic negativity $\mathcal E$ and the quantifiers $\mathcal E_4$, $\mathcal E_6$, $\mathcal E_8$ of a window of $n_A=n_B=5$ sites in the middle of the chain for (a) the TFIM at $h=0.85$ and (b) the XXZ chain at $\Delta=0.5$, as a function of the chain length $n$. The remaining $n-10$ sites are traced out. The dotted line marks $n=12$, where DMRG is checked against exact diagonalization. All larger sizes are obtained by DMRG.}
\label{fig:dmrg_scaling}
\end{figure}

\section*{Appendix D: Beyond exact diagonalization}\label{app:D}

For the TFIM and the XXZ chain we obtain the ground state as a matrix product state by two-site DMRG \cite{White1992,Schollwock2011}, implemented in \texttt{quimb} \cite{Gray2018Quimb}. The reduced state of the same window of $n_A=n_B=5$ adjacent sites is extracted from the matrix product state, and the other $n-10$ sites are traced out. The LMG model is not treated in this way, since its all-to-all coupling leaves no locality for a matrix product state to exploit. We validate the calculation in three ways. At $n=12$ a dense diagonalization of the Hamiltonian used for DMRG, written in the spin-operator convention of \texttt{quimb}, reproduces the exact results of Fig.~\ref{fig:manybody}. At every larger $n$ we run DMRG with bond dimensions $200$ and $350$ and accept a point only if the two energies agree to $10^{-6}$ and the two values of $\mu_2$ to $10^{-5}$. 

Finally, the inequalities $0<\mu_2<1$ and $\mu_2^2<\mu_3<\mu_2$ serve as a built-in test, since a bond dimension that is too small violates them before the energy shows any sign of it. With these criteria the TFIM at $h=0.85$, which is close to criticality and therefore the more demanding case, converges up to $n=32$ and fails to converge for $n=36$ to $40$. The XXZ chain at $\Delta=0.5$ converges up to $n=36$. We report no results beyond these sizes. Figure~\ref{fig:dmrg_scaling} shows that the negativity of the window decreases as the environment grows and that the ordering of Eq.~\eqref{eq:monotone} holds at every size.

\bibliography{refs}

\clearpage
\onecolumngrid
\setcounter{equation}{0}
\setcounter{figure}{0}
\renewcommand{\theequation}{S\arabic{equation}}
\renewcommand{\thefigure}{S\arabic{figure}}
\makeatletter
\renewcommand{\theHequation}{S.\arabic{equation}}
\renewcommand{\theHfigure}{S.\arabic{figure}}
\makeatother

\begin{center}
{\large\textbf{Supplemental Material}}
\end{center}

This Supplemental Material contains the proofs and the numerical details behind the main text. Section~\ref{sm:setup} recalls the notation. Section~\ref{sm:proof} proves Theorem~\ref{thm:main} on the optimal quantifier and discusses the scope of the optimality. Section~\ref{sm:small_K} derives the closed forms for $K=4$, $6$ and $8$. Section~\ref{sm:saturation} proves the monotonicity of the hierarchy, Proposition~\ref{prop:strict}, Theorem~\ref{thm:saturation} and Corollary~\ref{cor:stabilizer}. Section~\ref{sm:convergence} proves Theorem~\ref{thm:convergence}. Section~\ref{sm:manybody} describes the many-body models and the numerical setup, and Sec.~\ref{sm:higherK} the numerical evaluation of the hierarchy at high orders. Equations, figures, lemmas and sections of this Supplemental Material carry the prefix S. All other numbers refer to the main text.

\begin{center}\textbf{Contents}\end{center}
\begin{list}{}{\leftmargin=0.12\linewidth\rightmargin=0.12\linewidth\itemsep=0pt}
\item \hyperref[sm:setup]{S1. Setup and notation}\dotfill\pageref{sm:setup}
\item \hyperref[sm:proof]{S2. Proof of Theorem \ref{thm:main}}\dotfill\pageref{sm:proof}
\item \hyperref[sm:small_K]{S3. Closed forms for $K=4$, $6$ and $8$}\dotfill\pageref{sm:small_K}
\item \hyperref[sm:saturation]{S4. Monotonicity, strict improvement and saturation}\dotfill\pageref{sm:saturation}
\item \hyperref[sm:convergence]{S5. Proof of Theorem \ref{thm:convergence}}\dotfill\pageref{sm:convergence}
\item \hyperref[sm:manybody]{S6. Many-body models and numerical setup}\dotfill\pageref{sm:manybody}
\item \hyperref[sm:higherK]{S7. Higher orders: numerical evaluation and convergence}\dotfill\pageref{sm:higherK}
\end{list}

\smsection{Setup and notation}{sm:setup}

We collect the definitions of the main text. For a state $\rho$ on $\mathcal H_A\otimes\mathcal H_B$ with partial transpose $\rho^\Gamma$ and eigenvalues $\{\lambda_i\}$ of $\rho^\Gamma$, the logarithmic negativity and the PT moments are given in Eqs.~(1) and (2) of the main text. The normalized squared eigenvalues, the reduced moments and the measure $\nu$ are
\begin{align}
q_i=\frac{\lambda_i^2}{p_2}, \qquad \mu_k=\sum_iq_i^k=\frac{p_{2k}}{p_2^k}, \qquad \nu=\sum_{i:\,q_i>0}q_i\,\delta_{q_i}. \label{S:defs}
\end{align}
If a magnitude $|\lambda_i|$ occurs with multiplicity $g$, the corresponding point $x=q_i$ carries the weight $g\,x$. The number of points of $\nu$ equals the number $s$ of distinct nonzero magnitudes. With $h(x)=x^{-1/2}$ we have
\begin{align}
\int_0^1x^{k-1}d\nu=\mu_k, \qquad \int_0^1h\,d\nu=S=\sum_iq_i^{1/2}, \qquad \mathcal E(\rho)=\tfrac12\ln p_2+\ln S. \label{S:E_via_S}
\end{align}
For even $K\ge4$ we write $r=K/2$ and $\ell=\lfloor K/4\rfloor=\lfloor r/2\rfloor$. The set of admissible measures and the optimal value are
\begin{align}
\mathcal M_K&=\Big\{\sigma\ \text{probability measure on}\ [0,1]\ :\ \int_0^1x^{k-1}d\sigma=\mu_k,\ 1\le k\le r\Big\}, \label{S:MK}\\
S_K^-&=\inf_{\sigma\in\mathcal M_K}\int_0^1h\,d\sigma, \qquad \mathcal E_K(\rho)=\tfrac12\ln p_2+\ln S_K^-. \label{S:SK}
\end{align}
We set $h(0)=+\infty$, so a measure with weight at the origin has an infinite integral and plays no role in the infimum. Since $\nu\in\mathcal M_K$, we have $S_K^-\le S$ and $\mathcal E_K(\rho)\le\mathcal E(\rho)$.

\smsection{Proof of Theorem~\ref{thm:main}}{sm:proof}

\textit{Polynomial lower bounds.} Let $P(x)=\sum_{k=1}^{r}c_kx^{k-1}$ be a polynomial of degree at most $r-1$ with $P\le h$ on $[0,1]$. For every $\sigma\in\mathcal M_K$,
\begin{align}
\int_0^1h\,d\sigma\ \ge\ \int_0^1P\,d\sigma=\sum_{k=1}^{r}c_k\,\mu_k. \label{S:weak_duality}
\end{align}
The right-hand side depends on the moments only. Every such polynomial therefore gives a lower bound on $S_K^-$. If a measure $\sigma^\star\in\mathcal M_K$ is supported on points where $P=h$, then its integral equals the right-hand side of Eq.~\eqref{S:weak_duality}, and $\sigma^\star$ is a minimizer. The proof consists in finding such a pair.

\textit{Sign of the derivatives.} By induction,
\begin{align}
h^{(m)}(x)=(-1)^m\binom{2m}{m}\frac{m!}{4^m}\,x^{-m-1/2}, \qquad m=0,1,2,\dots \label{S:h_derivatives}
\end{align}
Hence $h^{(m)}>0$ on $(0,1]$ for even $m$ and $h^{(m)}<0$ for odd $m$.

We first treat the case in which $\nu$ has few points. The next lemma shows that the moments then leave no freedom.

\begin{smlemma}[Few magnitudes]
\label{lem:few}
If $2s\le r-1$, then $\mathcal M_K=\{\nu\}$. If $2s=r$, then $\nu$ is the $s$-point Gauss rule of $\mu_1,\dots,\mu_r$.
\end{smlemma}
\begin{proof}
Let $y_1,\dots,y_s$ be the points of $\nu$ and $\pi(x)=\prod_j(x-y_j)$. If $2s\le r-1$, the polynomial $\pi^2$ has degree at most $r-1$, so every $\sigma\in\mathcal M_K$ satisfies
\begin{align}
\int_0^1\pi^2\,d\sigma=\int_0^1\pi^2\,d\nu=0.
\end{align}
Thus $\sigma$ is supported on $\{y_j\}$. Its $s$ weights are fixed by the moments $\mu_1,\dots,\mu_s$ through a Vandermonde system, so $\sigma=\nu$. If $2s=r$, then $\int\pi\,x^{k}d\nu=0$ for all $k$, so $\pi$ is the monic polynomial of degree $s$ orthogonal to all lower degrees. Its roots are the Gauss nodes, and the weights are again fixed by the moments.
\end{proof}

\textit{Even $r=2\ell$ with $s>\ell$.} Since $\nu$ has more than $\ell$ points, the form $\langle f,g\rangle=\int fg\,d\nu$ is positive definite on polynomials of degree at most $\ell$. There is a unique monic polynomial $\pi_\ell$ of degree $\ell$ orthogonal to all lower degrees. It is fixed by $\mu_1,\dots,\mu_{2\ell}$, and its roots $x_1,\dots,x_\ell$ are real, distinct and lie strictly between the smallest and the largest point of $\nu$, hence in $(0,1)$ \cite{GautschiBook}. The Gauss rule
\begin{align}
\nu_K^\star=\sum_{j=1}^{\ell}w_j\,\delta_{x_j}, \qquad w_j>0, \label{S:gauss_rule}
\end{align}
integrates every polynomial of degree at most $2\ell-1$ exactly as $\nu$ does, so $\nu_K^\star\in\mathcal M_K$. Let $P$ be the Hermite interpolant of $h$ at the nodes, that is, the polynomial of degree at most $2\ell-1$ with $P(x_j)=h(x_j)$ and $P'(x_j)=h'(x_j)$. For every $x\in(0,1]$ there is a $\xi\in(0,1]$ with
\begin{align}
h(x)-P(x)=\frac{h^{(2\ell)}(\xi)}{(2\ell)!}\prod_{j=1}^{\ell}(x-x_j)^2. \label{S:hermite_gauss}
\end{align}
By Eq.~\eqref{S:h_derivatives} the right-hand side is positive away from the nodes. At $x=0$ the inequality $P<h$ holds trivially. So $P\le h$ on $[0,1]$ with equality only at the nodes, and
\begin{align}
\int_0^1h\,d\nu_K^\star=\sum_jw_jP(x_j)=\int_0^1P\,d\nu_K^\star=\sum_{k=1}^{r}c_k\,\mu_k\ \le\ S_K^-. \label{S:SK_gauss}
\end{align}
Together with $\nu_K^\star\in\mathcal M_K$ this gives $S_K^-=\int h\,d\nu_K^\star$. For uniqueness, let $\sigma\in\mathcal M_K$ be another minimizer. Then $\int(h-P)\,d\sigma=0$, so $\sigma$ is supported on the nodes, and its weights are fixed by the moments. Hence $\sigma=\nu_K^\star$.

\textit{Odd $r=2\ell+1$ with $s>\ell$.} Since $s>\ell\ge1$, the measure $\nu$ is not a single point and has no weight at $x=1$ (a point at $x=1$ means $\lambda_i^2=p_2$, which forces all other eigenvalues to vanish). The measure $(1-x)\,d\nu$ is then positive with $s>\ell$ points. Let $\tilde\pi_\ell$ be its monic orthogonal polynomial of degree $\ell$, fixed by $\mu_1,\dots,\mu_{2\ell+1}$, with roots $x_1,\dots,x_\ell$ in $(0,1)$. The Gauss-Radau rule
\begin{align}
\nu_K^\star=\sum_{j=1}^{\ell}w_j\,\delta_{x_j}+w_c\,\delta_1 \label{S:radau_rule}
\end{align}
integrates polynomials of degree at most $2\ell$ exactly as $\nu$ does \cite{GautschiBook}. Its weights are positive. For $w_j$ this follows from the Gauss rule of $(1-x)d\nu$, and for $w_c$ from applying the rule to $\tilde\pi_\ell^2$,
\begin{align}
w_c\,\tilde\pi_\ell(1)^2=\int_0^1\tilde\pi_\ell^2\,d\nu>0. \label{S:wc}
\end{align}
Let $P$ be the polynomial of degree at most $2\ell$ that agrees with $h$ in value and slope at the $x_j$ and in value at $x=1$. The remainder is
\begin{align}
h(x)-P(x)=\frac{h^{(2\ell+1)}(\xi)}{(2\ell+1)!}\,(x-1)\prod_{j=1}^{\ell}(x-x_j)^2. \label{S:hermite_radau}
\end{align}
Now $h^{(2\ell+1)}<0$ and $x-1\le0$, so the right-hand side is non-negative on $[0,1]$ and vanishes only at the nodes. A fixed node at $x=0$ would replace $x-1$ by $x\ge0$ and turn $P$ into an upper bound. The rest of the argument, including uniqueness, is the same as in the even case.

\textit{The case $s\le\ell$.} If $r$ is odd, or if $r$ is even and $s<\ell$, then $2s\le r-1$ and Lemma~\ref{lem:few} gives $\mathcal M_K=\{\nu\}$. If $r=2\ell$ and $s=\ell$, the Gauss rule exists by the argument of the even case and equals $\nu$ by Lemma~\ref{lem:few}. In both situations the unique minimizer is $\nu_K^\star=\nu$. This completes the proof of Theorem~\ref{thm:main}.

\textit{Scope of the optimality.} The theorem says that no function of $\mu_1,\dots,\mu_r$ that bounds $\int h\,d\sigma$ from below for all $\sigma\in\mathcal M_K$ can exceed $S_K^-$. The set $\mathcal M_K$ is larger than the set of measures that come from a PT spectrum. Three properties of a physical spectrum are not used. First, the weight of each point of $\nu$ is an integer multiple of its position, see below Eq.~\eqref{S:defs}, and the number of points is bounded by the dimension. Second, the odd moments $p_3,p_5,\dots$ depend on the signs of the $\lambda_i$ and do not enter. Third, the normalization $p_1=\sum_i\lambda_i=1$ is not imposed. The last property has a simple consequence,
\begin{align}
\sum_i|\lambda_i|\ \ge\ \sum_i\lambda_i=1 \quad\Longrightarrow\quad \mathcal E(\rho)\ge0, \label{S:nonneg}
\end{align}
which gives the certified value $\mathcal E_K^+=\max\{\mathcal E_K,0\}$ of Eq.~(9). A raw value $\mathcal E_K<0$ occurs when the even moments are also compatible with a measure for which $\sum_i|\lambda_i|<1$, which no physical spectrum allows. The question whether the first two properties can be used to raise the bound above $\mathcal E_K^+$, is open. They are convex or combinatorial side conditions, and the convex ones can be added to a semidefinite formulation of the moment problem \cite{BertsimasPopescu2005,Lasserre2009}.

\smsection{Closed forms for \texorpdfstring{$K=4$, $6$ and $8$}{K=4, 6 and 8}}{sm:small_K}

The closed forms need the reduced moments to satisfy a few inequalities. The following lemma provides them and also describes the admissible region used in Appendix~C of the End Matter.

\begin{smlemma}[Log-convexity]
\label{lem:logconvex}
For the reduced moments of any PT spectrum and every $k\ge2$,
\begin{align}
\mu_k\le\mu_{k-1}\le1, \qquad \mu_k^2\le\mu_{k-1}\,\mu_{k+1}. \label{S:logconvex}
\end{align}
The first inequality is an equality only if $\nu=\delta_1$, and the second only if $\nu$ is a single point.
\end{smlemma}
\begin{proof}
Since $0<q_i\le1$ on the support, $q_i^k\le q_i^{k-1}$ with equality only for $q_i=1$. Summation gives the first claim. The second is the Cauchy-Schwarz inequality for the vectors with entries $q_i^{(k-1)/2}$ and $q_i^{(k+1)/2}$,
\begin{align}
\mu_k^2=\Big(\sum_iq_i^{(k-1)/2}q_i^{(k+1)/2}\Big)^2\le\Big(\sum_iq_i^{k-1}\Big)\Big(\sum_iq_i^{k+1}\Big).
\end{align}
Equality requires the two vectors to be proportional, so that all nonzero $q_i$ are equal.
\end{proof}

\textit{$K=4$.} Here $r=2$ and $\ell=1$. A single node with unit weight reproduces $\mu_1=1$, and $\mu_2$ fixes its position,
\begin{align}
\nu_4^\star=\delta_{\mu_2}, \qquad S_4^-=\mu_2^{-1/2}, \qquad \mathcal E_4(\rho)=\frac12\ln\frac{p_2^3}{p_4}. \label{S:E4}
\end{align}
This is the $p_4$ bound of Ref.~\cite{TarabungaHaug2025}.

\textit{$K=6$.} Here $r=3$ and $\ell=1$, so $\nu_6^\star=w_a\delta_a+w_c\delta_1$ with
\begin{align}
w_a+w_c=1, \qquad w_a\,a+w_c=\mu_2, \qquad w_a\,a^2+w_c=\mu_3. \label{S:radau_system}
\end{align}
Subtracting consecutive equations gives $w_a(1-a)=1-\mu_2$ and $w_a\,a(1-a)=\mu_2-\mu_3$, and therefore
\begin{align}
a=\frac{\mu_2-\mu_3}{1-\mu_2}, \qquad w_a=\frac{(1-\mu_2)^2}{1-2\mu_2+\mu_3}, \qquad w_c=\frac{\mu_3-\mu_2^2}{1-2\mu_2+\mu_3}. \label{S:radau_weights}
\end{align}
The denominator and the numerator of $w_c$ are
\begin{align}
1-2\mu_2+\mu_3=\int_0^1(1-x)^2d\nu, \qquad \mu_3-\mu_2^2=\int_0^1(x-\mu_2)^2d\nu. \label{S:denom_identity}
\end{align}
By Lemma~\ref{lem:logconvex} both are positive and $0<a<1$ whenever $\nu$ has more than one point. The result is
\begin{align}
S_6^-=w_a\,a^{-1/2}+w_c=\frac{(1-\mu_2)^{5/2}}{(1-2\mu_2+\mu_3)\sqrt{\mu_2-\mu_3}}+\frac{\mu_3-\mu_2^2}{1-2\mu_2+\mu_3}. \label{S:S6}
\end{align}
If $\nu$ is a single point, then $w_c=0$ and $a=\mu_2$, and Eq.~\eqref{S:S6} reduces to Eq.~\eqref{S:E4}. The admissible region for $K=6$ is thus $0<\mu_2<1$ and $\mu_2^2<\mu_3<\mu_2$.

\textit{$K=8$.} Here $r=4$ and $\ell=2$. The nodes are the roots of $\pi_2(x)=x^2-\alpha x+\beta$, which is orthogonal to $1$ and $x$,
\begin{align}
\mu_3-\alpha\mu_2+\beta=0, \qquad \mu_4-\alpha\mu_3+\beta\mu_2=0. \label{S:orth8}
\end{align}
This gives
\begin{align}
\alpha=\frac{\mu_4-\mu_2\mu_3}{\mu_3-\mu_2^2}, \qquad \beta=\alpha\mu_2-\mu_3, \qquad x_{1,2}=\frac{\alpha\pm\sqrt{\alpha^2-4\beta}}{2}, \label{S:alpha_beta}
\end{align}
and the weights follow from $w_1+w_2=1$ and $w_1x_1+w_2x_2=\mu_2$,
\begin{align}
w_1=\frac{\mu_2-x_2}{x_1-x_2}, \qquad w_2=\frac{x_1-\mu_2}{x_1-x_2}, \qquad S_8^-=w_1\,x_1^{-1/2}+w_2\,x_2^{-1/2}. \label{S:S8}
\end{align}
For $s\ge2$ the nodes lie in $(0,1)$ and the weights are positive (Sec.~\ref{sm:proof}). For $s=1$ one has $\mu_3=\mu_2^2$ and the bound is already exact at $K=4$. Close to this case Eq.~\eqref{S:alpha_beta} is a ratio of small numbers and should be evaluated with care.

\textit{Two-qubit spectra.} For the Werner state of Eq.~(16) the partial transpose has the eigenvalue $(1+p)/4$ with multiplicity three and $(1-3p)/4$ with multiplicity one. For $p>1/3$ this gives $\mathcal E=\ln[(1+3p)/2]$. There are two magnitudes for $0<p<1$, except at $p=1/3$ where one of them is zero. For the pure state of Eq.~(17) the eigenvalues are $\alpha^2$, $1-\alpha^2$ and $\pm\alpha\sqrt{1-\alpha^2}$, with three magnitudes unless $\alpha^2\in\{0,1/2,1\}$.

\smsection{Monotonicity, strict improvement and saturation}{sm:saturation}

We first record the ordering of Eq.~(14), on which the other results rely.

\begin{smprop}[Monotonicity]
\label{prop:monotone}
$S_K^-\le S_{K+2}^-\le S$ for every even $K\ge4$.
\end{smprop}
\begin{proof}
One more moment is one more linear constraint, so $\mathcal M_{K+2}\subseteq\mathcal M_K$. The infimum over a smaller set is not smaller. The second inequality holds since $\nu\in\mathcal M_{K+2}$.
\end{proof}

\textit{Proof of Theorem~\ref{thm:saturation}.} If $s\le\ell$, Theorem~\ref{thm:main} gives $\nu_K^\star=\nu$ and hence $\mathcal E_K=\mathcal E$. Let $s>\ell$. For even $r$ the minimizer $\nu_K^\star$ has $\ell<s$ points. For odd $r$ it has positive weight at $x=1$ by Eq.~\eqref{S:wc}, while $\nu$ has none. In both cases $\nu\ne\nu_K^\star$. Since $\nu\in\mathcal M_K$ and the minimizer is unique, $S>S_K^-$ and $\mathcal E_K<\mathcal E$. \hfill$\square$

\textit{Proof of Proposition~\ref{prop:strict}.} If $\mathcal E_K=\mathcal E$, Proposition~\ref{prop:monotone} gives $\mathcal E_{K+2}=\mathcal E_K$. Let $\mathcal E_K<\mathcal E$, so that $s>\ell$ by Theorem~\ref{thm:saturation}. We show that $\nu_K^\star$ does not reproduce the next moment $\mu_{r+1}$. For $r=2\ell$ we write $x^{2\ell}=\pi_\ell(x)^2+R(x)$ with $R$ of degree at most $2\ell-1$. Both measures integrate $R$ equally, and $\pi_\ell$ vanishes on the nodes, so
\begin{align}
\int_0^1x^{2\ell}d\nu-\int_0^1x^{2\ell}d\nu_K^\star=\int_0^1\pi_\ell^2\,d\nu>0. \label{S:next_even}
\end{align}
The integral is positive because $\nu$ has more than $\ell$ points. For $r=2\ell+1$ the polynomial $(1-x)\tilde\pi_\ell(x)^2$ has degree $2\ell+1$ and leading coefficient $-1$, and it vanishes on all nodes of the Gauss-Radau rule. In the same way
\begin{align}
\int_0^1x^{2\ell+1}d\nu_K^\star-\int_0^1x^{2\ell+1}d\nu=\int_0^1(1-x)\,\tilde\pi_\ell^2\,d\nu>0. \label{S:next_odd}
\end{align}
Hence $\nu_K^\star\notin\mathcal M_{K+2}$. By Theorem~\ref{thm:main} the infimum over $\mathcal M_{K+2}$ is attained by a measure $\nu_{K+2}^\star$, which belongs to $\mathcal M_K$ and differs from $\nu_K^\star$. Uniqueness of the minimizer over $\mathcal M_K$ gives $S_{K+2}^->S_K^-$. \hfill$\square$

\textit{Proof of Corollary~\ref{cor:stabilizer}.} Let $|\psi\rangle=\sum_ic_i|a_i\rangle|b_i\rangle$ be the Schmidt decomposition of a pure state. The eigenvalues of its partial transpose are $c_i^2$ for every $i$ and $\pm c_ic_j$ for every pair $i<j$ \cite{VidalWerner2002}. For a stabilizer state the reduced state of either part is proportional to a projector of rank $R=2^k$, with $k$ an integer \cite{HeinEisertBriegel2004,Gottesman1997,Gottesman1998Theory,Gottesman1998Heisenberg,AaronsonGottesman2004}. Then $c_i=R^{-1/2}$ and all nonzero eigenvalues are $\pm1/R$, so that
\begin{align}
p_2=1, \qquad p_3=p_4=\frac{1}{R^2}, \qquad \mathcal E_4=\frac12\ln\frac{p_2^3}{p_4}=\ln R=\mathcal E. \label{S:stab}
\end{align}
The spectrum is $\lambda_i=\pm\sqrt{p_3}$, in agreement with Ref.~\cite{TarabungaHaug2025}. There is one magnitude, so Theorem~\ref{thm:saturation} applies with $K=4$, and Proposition~\ref{prop:monotone} gives equality at all higher orders. \hfill$\square$

\smsection{Proof of Theorem~\ref{thm:convergence}}{sm:convergence}

The sequence $S_K^-$ is non-decreasing and bounded by $S$, so it converges. We show that the limit is $S$. The probability measures on $[0,1]$ form a compact set in the weak topology, and each moment is a continuous function of the measure. The sets $\mathcal M_K$ are therefore compact, and they are nested. A measure in their intersection has all moments of $\nu$. Since polynomials are dense in the continuous functions on $[0,1]$, it equals $\nu$,
\begin{align}
\bigcap_K\mathcal M_K=\{\nu\}. \label{S:intersection}
\end{align}
It follows that the minimizers $\nu_K^\star$ converge weakly to $\nu$. Otherwise a subsequence would stay outside a neighborhood of $\nu$, and by compactness it would have a limit point in every $\mathcal M_K$ different from $\nu$, in contradiction with Eq.~\eqref{S:intersection}. The function $h$ is the increasing limit of the bounded continuous functions $h_N=\min(h,N)$, so
\begin{align}
\int_0^1h\,d\sigma=\sup_N\int_0^1h_N\,d\sigma \label{S:lsc}
\end{align}
is lower semicontinuous in $\sigma$. Therefore
\begin{align}
S=\int_0^1h\,d\nu\ \le\ \liminf_{K\to\infty}\int_0^1h\,d\nu_K^\star=\lim_{K\to\infty}S_K^-\ \le\ S, \label{S:limit}
\end{align}
which proves the theorem. \hfill$\square$

\smsection{Many-body models and numerical setup}{sm:manybody}

\textit{Geometry and numerical procedure.} All ground-state results of the main text are for $n=12$ spins with open boundary conditions. The chain is divided into three parts. The blocks $A$ and $B$ have $n_A=n_B=5$ sites each and are adjacent in the middle of the chain. The part $C$ consists of the two remaining sites, one at each end. This is the connected geometry of Ref.~\cite{TarabungaHaug2025}, which we adopt for direct comparison with the two-moment bound. The ground state $|\psi\rangle$ is obtained by sparse Lanczos diagonalization. The iteration is continued until the relative residual of the energy is below $10^{-12}$, and the Lanczos vectors are re-orthogonalized at every step. The reduced state
\begin{align}
\rho_{AB}=\mathrm{tr}_C\,|\psi\rangle\langle\psi| \label{S:rhoAB}
\end{align}
is a matrix of dimension $2^{10}=1024$. We form it by writing $|\psi\rangle$ as a tensor with one index per site and contracting the indices of $C$ with their conjugates, so that the density matrix of the full chain is never built. The partial transpose on $B$ is diagonalized in full, without truncation. From its spectrum we compute $\mathcal E$ and the moments $p_2,\dots,p_8$, and from these $\mathcal E_4$, $\mathcal E_6$, $\mathcal E_8$ with the formulas of Sec.~\ref{sm:small_K}. At every point we check that $\sum_i\lambda_i=1$ to double precision and that $p_2$ agrees with the purity of $\rho_{AB}$, computed without the partial transpose, to twelve digits. No noise and no temperature enter these scans.

The three models were chosen to cover different situations with the same geometry. They are an isolated critical point with a discrete symmetry, a critical phase with a continuous symmetry, and a model without any spatial structure. A feature seen in all three can then be attributed to the PT spectrum and not to a particular Hamiltonian. The negativity of many-body states has been studied widely, in field theory \cite{Calabrese2012Negativity,Calabrese2014Finite,RuggieroAlbaCalabrese2016b}, in spin and harmonic chains \cite{RuggieroAlbaCalabrese2016a,Eisler2014}, and in topologically ordered phases \cite{Castelnovo2013,LeeVidal2013,WenChangRyu2016,LuHsiehGrover2020}.

\textit{Transverse-field Ising model.} The Hamiltonian is
\begin{align}
H_{\mathrm{TFIM}}=-\sum_{k=1}^{n-1}\sigma^x_k\sigma^x_{k+1}-h\sum_{k=1}^{n}\sigma^z_k. \label{S:H_tfim}
\end{align}
The first term favors spins aligned along $x$ and the field $h$ favors alignment along $z$. The model is invariant under the simultaneous flip $\sigma^x_k\to-\sigma^x_k$ of all spins, a $\mathbb Z_2$ symmetry generated by $\prod_k\sigma^z_k$. It is the standard example of a continuous quantum phase transition and can be solved for any size by a Jordan-Wigner mapping to free fermions. For the infinite chain the fermions have the dispersion
\begin{align}
\varepsilon_k=2\sqrt{1+h^2-2h\cos k}, \label{S:tfim_disp}
\end{align}
so the gap is $2|1-h|$ and closes at $h=1$. For $h<1$ the chain is ferromagnetically ordered, with a spontaneous magnetization along $x$ that vanishes as $(1-h^2)^{1/8}$ at the transition. For $h>1$ it is a paramagnet. The correlation length diverges as $h\to1$ from either side. The model is realized in arrays of Rydberg atoms, in trapped-ion chains and in superconducting circuits, which are the platforms where PT moments can be measured. The critical point is described by a conformal field theory with central charge $1/2$, and the entanglement entropy of a block grows logarithmically with its size there \cite{VidalLatorreRicoKitaev2003,CalabreseCardy2004}. For a finite chain the ground state at small $h$ is the symmetric combination of the two ordered states,
\begin{align}
|\psi\rangle\ \to\ \frac{1}{\sqrt2}\big(|{+}{+}\cdots{+}\rangle+|{-}{-}\cdots{-}\rangle\big) \quad (h\to0), \label{S:cat}
\end{align}
with $|\pm\rangle$ the eigenstates of $\sigma^x$. Tracing out $C$ turns this cat state into an equal mixture of two product states, which is separable. For large $h$ the ground state approaches the product state polarized along $z$. Entanglement between $A$ and $B$ is therefore largest in between, close to the critical point, where many Schmidt levels are populated. This is the regime that is hardest for a moment bound, since the PT spectrum then has many magnitudes of comparable size. We scan $h\in[0,2]$ at $26$ equally spaced points.

\textit{XXZ chain.} The Hamiltonian is
\begin{align}
H_{\mathrm{XXZ}}=-\sum_{k=1}^{n-1}\big(\sigma^x_k\sigma^x_{k+1}+\sigma^y_k\sigma^y_{k+1}+\Delta\,\sigma^z_k\sigma^z_{k+1}\big). \label{S:H_xxz}
\end{align}
It conserves the total magnetization $\sum_k\sigma^z_k$, a continuous $U(1)$ symmetry in place of the discrete symmetry of the TFIM. The first two terms exchange neighboring up and down spins, and $\Delta$ sets the strength of the interaction along $z$. A rotation of every second spin by $\pi$ about the $z$ axis changes the sign of the exchange terms and maps Eq.~\eqref{S:H_xxz} to the common antiferromagnetic form with anisotropy $-\Delta$. In our convention the chain is therefore a gapped N\'eel antiferromagnet for $\Delta<-1$ and a ferromagnet for $\Delta>1$. In between, for $-1\le\Delta<1$, it is gapless and described by a Luttinger liquid with central charge one and exponents that vary continuously with $\Delta$. The model is solvable by the Bethe ansatz for all $\Delta$. At $\Delta=-1$ it is the isotropic Heisenberg antiferromagnet, and the gap of the N\'eel phase opens beyond this point exponentially slowly. Deep in the N\'eel phase the ground state of a finite chain approaches the superposition of the two N\'eel states, which is again a cat state. Its entanglement between $A$ and $B$ is small once $C$ is traced out, which explains the decrease of all quantifiers towards $\Delta=-4$ in Fig.~2(b). At $\Delta=1$ the model is the isotropic ferromagnet, whose ground states form a degenerate multiplet of total spin $n/2$. For $\Delta>1$ the ground states are the two fully polarized product states. The approach to this degenerate point is the origin of the kink near $\Delta=1$. The critical behavior thus extends over a whole interval. This allows us to test whether the wide gap between $\mathcal E_4$ and $\mathcal E$ at the Ising critical point is a feature of one point or of gapless chains in general. We scan $\Delta\in[-4,1]$ at $26$ points, the range of Ref.~\cite{TarabungaHaug2025}.

\textit{Lipkin-Meshkov-Glick model.} The Hamiltonian is
\begin{align}
H_{\mathrm{LMG}}=-\frac{1}{n-1}\sum_{i<j}\sigma^x_i\sigma^x_j-h\sum_{k=1}^{n}\sigma^z_k. \label{S:H_lmg}
\end{align}
Every pair of spins interacts with the same strength, and the factor $1/(n-1)$ keeps the energy per spin finite. The model was introduced in nuclear physics \cite{Lipkin1965} and has become a standard example for entanglement in systems with infinite range \cite{LatorreOrusRicoVidal2005,AmicoFazioOsterlohVedral2008}. Such interactions arise for atoms coupled to a common cavity mode and for trapped ions coupled by a shared phonon mode. In terms of the collective spin $J_a=\tfrac12\sum_k\sigma^a_k$ the Hamiltonian reads
\begin{align}
H_{\mathrm{LMG}}=-\frac{2}{n-1}\Big(J_x^2-\frac n4\Big)-2h\,J_z. \label{S:H_lmg_J}
\end{align}
It commutes with the total spin, and the ground state lies in the symmetric sector of dimension $n+1$. In the limit of large $n$ a classical spin pointing at an angle $\theta$ from the $z$ axis has the energy per spin $-\tfrac12\sin^2\theta-h\cos\theta$, which is minimal at $\cos\theta=h$ for $h<1$ and at $\theta=0$ for $h>1$. The transition is at $h=1$. It is of second order with mean-field exponents, and at the transition the entanglement entropy of a block grows logarithmically with the size \cite{LatorreOrusRicoVidal2005}. Like the TFIM, the model is invariant under the flip generated by $\prod_k\sigma^z_k$. At $h=0$ the two states with all spins along $+x$ or $-x$ are degenerate, and a small field selects their symmetric combination, the same cat state as in Eq.~\eqref{S:cat}. For $n=12$ the entanglement of our partition peaks below the transition, at $h\approx0.73$, a finite-size shift that disappears slowly with $n$. Because of the permutation symmetry the PT spectrum has far fewer distinct magnitudes than for the two chains. The model therefore tests whether the behavior of the hierarchy is tied to locality or only to the spectrum. We scan $h\in[0,3]$.

\textit{Results.} The ordering $\mathcal E_4\le\mathcal E_6\le\mathcal E_8\le\mathcal E$ holds at every scanned point (Fig.~2). For the TFIM at $h=0.85$ we find
\begin{align}
\mathcal E=0.542, \qquad \mathcal E_4=0.245, \qquad \mathcal E_6=0.290, \qquad \mathcal E_8=0.380. \label{S:tfim_numbers}
\end{align}
The fraction $(\mathcal E_8-\mathcal E_4)/(\mathcal E-\mathcal E_4)$ of the gap recovered by $\mathcal E_8$ lies between $0.45$ and $0.60$ for $0.7\le h\le1$. For $h\to0$ the ground state approaches the cat state, the reduced state of $AB$ is a separable mixture, and all quantifiers vanish. In the paramagnet the two lowest orders certify little, for instance $\mathcal E_4=0.004$ at $h=2$, where $\mathcal E=0.238$ and $\mathcal E_8=0.195$. In the XXZ chain the gap between $\mathcal E_4$ and $\mathcal E$ is wide over the whole critical region. For $\Delta\lesssim-2.5$ the raw values of $\mathcal E_4$ and $\mathcal E_6$ are negative, while $\mathcal E_8$ stays positive. At $\Delta=-4$ we find $\mathcal E=0.452$ and $\mathcal E_8=0.317$. The last point of the scan, $\Delta=1$, is the isotropic ferromagnet. Its ground state is degenerate, the Lanczos solver returns one state of the multiplet, and the values at this point are not unique. For the LMG model the quantifiers rise from zero at $h=0$, peak near $h\approx0.73$ and decay at large field.

\textit{Thermal state.} For Fig.~3 we take the TFIM at $h=0.9$ on $n=10$ sites and form the Gibbs state $\rho(T)=e^{-H/T}/Z$ by full diagonalization, for temperatures from $T=10^{-3}$ to $T=12$ on a logarithmic grid. Two sites at the ends are traced out, and the partial transpose is taken on one of the two remaining blocks of four sites. At the lowest temperature the ground-state values are recovered. The negativity then decreases and vanishes at $T\approx1.21$, a value located by bisection between two grid points, as expected for thermal entanglement \cite{ArnesenBoseVedral2001,HartCastelnovo2018}. The raw $\mathcal E_4$, $\mathcal E_6$, $\mathcal E_8$ cross zero at lower temperatures and are negative above the threshold, where Eq.~\eqref{S:nonneg} applies.

\smsection{Higher orders: numerical evaluation and convergence}{sm:higherK}

\textit{Evaluation for general $K$.} The nodes of the Gauss rule are the eigenvalues of the Jacobi matrix, the symmetric tridiagonal matrix of size $\ell$ built from the recurrence coefficients of the orthogonal polynomials of the moment sequence. The weights are the squared first components of its normalized eigenvectors. The Gauss-Radau rule is obtained from a Jacobi matrix of size $\ell+1$ whose last diagonal entry is modified so that $x=1$ is an eigenvalue \cite{GautschiBook}. The recurrence coefficients follow from the moments $\mu_1,\dots,\mu_r$ by the Cholesky factorization of their Hankel matrix. This step is ill-conditioned for large $K$, and we carry it out in arithmetic with $100$ decimal digits. For our many-body spectra it remains reliable up to $K\approx60$ to $80$.

\textit{When does the hierarchy converge?} Three statements of the main text answer this question. The bound is exact from the order $K_{\rm sat}=4s$ on, with $s$ the number of distinct nonzero magnitudes (Theorem~\ref{thm:saturation}). Below this order every step is a strict improvement (Proposition~\ref{prop:strict}). The limit $K\to\infty$ gives $\mathcal E$ for any $s$ (Theorem~\ref{thm:convergence}). Since $s$ is at most the dimension of $\mathcal H_A\otimes\mathcal H_B$, the hierarchy of a finite-dimensional state always terminates. The speed of the approach before saturation is not fixed by these results and depends on the spectrum. The gap comes from the smallest magnitudes, where $h$ is large and steep, so a spectrum with a long tail of small eigenvalues converges slowly. The convergence curves for two-qubit and many-body states, computed with the procedure described above, are shown in Figs.~4 and 5 of the main text (Appendix~B of the End Matter).


\end{document}